\documentclass[12pt]{article}
\usepackage{jheppub}
\usepackage{graphicx, amsmath, amsthm, amssymb}

\theoremstyle{plain} 

\newtheorem{theorem}{Theorem}[section]
\newtheorem{lemma}[theorem]{Lemma}
\newtheorem{claim}[theorem]{Claim}

\usepackage[
    backend=biber,       
    style=alphabetic,       
    sorting=none,
    maxbibnames=99       
]{biblatex}
\usepackage{mathtools}
\DeclarePairedDelimiter{\abs}{\lvert}{\rvert}
\DeclarePairedDelimiter{\norm}{\lVert}{\rVert}
\DeclarePairedDelimiter{\parens}{\lparen}{\rparen}
\DeclarePairedDelimiter{\bracks}{\lbrack}{\rbrack}
\DeclarePairedDelimiter{\innerprod}{\langle}{\rangle}
\newcommand{\ZZ}{\mathbb{Z}}
\newcommand{\CC}{\mathbb{C}}
\newcommand{\PP}{\mathbb{P}}
\DeclareMathOperator{\PGL}{PGL}
\DeclareMathOperator{\SU}{SU}
\DeclareMathOperator{\Aut}{Aut}
\DeclareMathOperator{\Isom}{Isom}
\DeclareMathOperator{\Ric}{Ric}
\DeclareMathOperator{\im}{im}

\title{The Sharp Edges of Calabi--Yau Manifolds:\\ Designing Symmetric Models for Ricci-flat Metrics}

\author{Viktor Mirjani\'c and Challenger Mishra}

\affiliation{
Department of Computer Science and Technology, \\
William Gates Building, 15 JJ Thomson Avenue, \\
Cambridge CB3 0FD}
\emailAdd{vvm22@cam.ac.uk}
\emailAdd{cm2099@cam.ac.uk}

\abstract{Computing Ricci-flat metrics on Calabi--Yau manifolds is challenging since no closed-form solutions are known. However, these computations are needed in order to make physical predictions in heterotic string theory, such as the masses of quarks and Yukawa couplings. In this manuscript, we present an overview of relevant literature for learning about Calabi--Yau manifolds, as ML researchers often face a steep learning curve when entering the field. Furthermore, we survey the impact of the manifold's symmetries on machine learning approximations to these flat metrics. We also characterise the isometries of Ricci-flat metrics, a result frequently omitted or used without proof. Then, we address symmetry breaking in point sampling and introduce a novel formula for computing volume ratios on general CICY manifolds. We conclude by presenting a new symmetry-aware model built using graph neural networks that avoids pathological behaviour witnessed in some other models.
}

\begin{document}

\maketitle

\section{Introduction}

One of the most promising approaches that connect string theory to four-dimensional physics is to posit that the extra dimensions of space live on a manifold that is then compactified to recreate the known universe. Thus, the phenomenological study of string theory is directly tied to differential geometry of complex manifolds. 

The main bottleneck in computing four-dimensional physics is the non-constructive nature of the Calabi--Yau theorem, the core theorem that governs the underlying geometry. This is where AI for science can offer a path forward by integrating high-precision computation with existing theory.

In this paper we investigate the use of graph neural networks (GNNs) for computing Ricci-flat metrics on Calabi--Yau manifolds. Along the way, we systematically address various subtleties, or the ``sharp edges'', in the machine learning pipeline. Specifically, we focus on symmetries of the manifold, ways in which they can be incorporated into the model, and situations where they need to be broken.

In Section \ref{sec:2}, we will summarise relevant literature that will help an ML researcher joining the field. Then, in Sections \ref{sec:3} and \ref{sec:4} we will discuss issues regarding symmetries and data generation that might arise when deploying these models. We will also give a novel formula for a volume ratio on complete intersection Calabi--Yau manifolds (CICYs) that is of theoretical interest. Finally, in Section \ref{sec:5} we introduce a novel architecture for metric learning based on GNNs, and comment on pathological scenarios where minimising training loss makes the metric diverge.

\section{String Literature Landscape}\label{sec:2}

A special class of manifolds of interest are K\"ahler manifolds, characterised by the interplay of their three structures: complex, Riemannian, and symplectic. The first means that these manifolds are not just real, but the coordinates have a well-defined complex structure. The second means that the manifold is equipped with a metric, giving it a notion of distance, angles, and volumes. Finally, symplectic structure ties everything to cohomology, which in turn introduces various topological structures. For example, all K\"ahler manifolds have a well-defined volume form $\Omega$, and thus a total non-zero volume.

A particular subclass of K\"ahler manifolds that yields the vacuum Einstein field equations when compactified are Calabi--Yau manifolds. Due to their structure, they can be defined in multiple equivalent ways. They are manifolds with a \emph{holomorphic} volume form $\Omega$, and with holonomy $\SU(n)$. When $n=3$, large datasets of Calabi--Yau manifolds exist, in the form of 7890 CICYs \cite{candelas_complete_1988}, or the Kreuzer--Skarke list \cite{kreuzer_complete_2000}. 

Perhaps most relevantly for machine learning, Calabi--Yau manifolds possess a Ricci-flat metric by Yau's theorem. Because no explicit construction exists, the most typical application of ML is to numerically approximate this flat metric. To this end, multiple packages, such as \emph{cymetric} \cite{larfors_learning_2021} and \emph{cymyc} \cite{berglund_cymyc_2024}, have been developed. Donaldson's algorithm \cite{douglas_numerical_2008,donaldson_scalar_2001,donaldson_numerical_2005} is an alternative that benefits from provable convergence guarantees, but converges orders of magnitude slower than ML. Yet other approaches utilise energy functionals \cite{headrick_energy_2010} to achieve low loss, but have not been generalised beyond highly symmetric manifolds. However, flat metric learning has been thoroughly explored over the years \cite{ashmore_machine_2020,ashmore_calabiyau_2023,jejjala_neural_2022}, so any research level work will require some knowledge of theory, as well as applying novel techniques such as interpretable ML \cite{mirjanic_symbolic_2025,lee_approximate_2025,constantin_calabi-yau_2026}.

For foundational concepts, \textcite{hatcher_algebraic_2002} is an excellent introduction to topology and cohomology, and \textcite{nakahara_geometry_2003} covers the foundational background for Calabi--Yau manifolds. Next, one can move towards Candelas and de la Ossa's lecture notes on complex manifolds \cite{candelas_lectures_2000}, a standard, albeit terse, reference in this field. 

There are many other helpful manuscripts such as \textcite{anderson_lectures_2023,bouchard_lectures_2007,hubsch_calabi-yau_2024}, and this is by no means an exhaustive list. Looking at code implementations can also be a good way to learn about computations of Yukawa couplings \cite{butbaia_physical_2024}, for example, and many papers offer excellent surveys of specialised topics \cite{ek_calabi-yau_2024,gross_large_2000}.

Thus, many new research directions open up, such as solving Yang--Mills equations in non-Abelian settings \cite{mishra_hermitian_2025}, or learning metrics beyond Calabi--Yau manifolds \cite{anderson_moduli-dependent_2021}, or learning on higher-dimensional spaces \cite{aggarwal_machine_2024}. Nevertheless, the Calabi--Yau manifolds are far from solved, and still offer plenty of challenges.

\section{On Calabi--Yau Symmetries}\label{sec:3}

Knowing the symmetry group of a manifold is very useful both for theory and for numerical experiments. For example, when learning the Ricci-flat metric numerically, we can encode symmetry constraints into the model, giving it inductive biases that simplify the training process. 

In fact, since the role of Calabi--Yau manifolds in string theory is to be compactified in order to reconstruct the Standard Model, studying their symmetries is not just a theoretical exercise or an aid to computation, but a fundamental requirement for their downstream applications \cite{witten_symmetry_1985}.

That being said, symmetries of Calabi--Yau manifolds can refer to slightly different things depending on context. Firstly, they can refer to automorphisms of the manifolds. A classical result is that automorphism groups of smooth hypersurfaces of degree at least three are always finite \cite{matsumura_automorphisms_1963}.

Even if finite, a manifold with a concrete choice of moduli might have a significantly large automorphism group. A Fermat quintic in $\PP^4$, defined by
\begin{equation*}
    \label{eq:fermat}
    Z_0^5 + Z_1^5 + Z_2^5 + Z_3^5 + Z_4^5 = 0,
\end{equation*}
has $\ZZ_2$ conjugation symmetry, $S_5$ permutation symmetry, and a large $\ZZ_5^4$ toric symmetry.\footnote{Note that the toric group is not $\ZZ_5^5$, because the last $\ZZ_5$ is contained in the projective $\CC^*$ symmetry.} Thus its total holomorphic automorphism group is $\Aut(X)= \ZZ_5^4 \rtimes S_5$, with order $\abs{\Aut(X)}=75{,}000$. Including the non-holomorphic $\ZZ_2$ group, the order would be twice as big.

Finding all automorphisms is difficult in general, and in practice we prefer to work with ones that are linear in ambient space coordinates. For the quintic, those are a subset of $\PGL$, and it just so happens that they include all automorphisms on the manifold. In general, however, it could happen that multiple CICY configurations give rise to the same manifold, and that linear symmetries in one representation correspond to non-linear actions in the other.

\subsection{Free Symmetries}

Some of the automorphisms will be free, while others will have fixed points on the manifold. This distinction is important for phenomenology, because one can quotient by the free group $\Gamma \le \Aut(X)$ to obtain a new smooth Calabi--Yau. Crucially, this new manifold will have different Hodge numbers than the original, making this a model way to obtain new Calabi--Yau manifolds from the previous ones \cite{candelas_vacuum_1985}.

For this reason, the free groups of 7890 CICYs have been exhaustively classified \cite{braun_free_2011}, and similar attempts have been made for the Kreuzer--Skarke list \cite{braun_discrete_2018}.

There is a delicate balance between breaking symmetries to obtain manifolds with smaller Hodge numbers, and preserving them to influence compactification. For example, a residual $\ZZ_2$ symmetry from the $R$-symmetry breaking is beneficial, because it will manifest as matter parity and ensure stability of the proton \cite{ibanez_discrete_1992}. On the other hand, restricting to symmetric manifolds will constrain the moduli space and limit the number of parameters that can be fine-tuned to fit the Standard Model.

A classification of all discrete symmetries is still not complete. Approaches that worked for free symmetries relied on limiting the group order, and this is hard to translate to the non-free case, where the groups can be much larger. Therefore, only special cases such as non-free linear symmetries have been classified \cite{lukas_discrete_2020}.

Furthermore, restricting the moduli space of Calabi--Yau manifolds can enlarge the symmetry group. A general hypersurface in $\PP^4$ does not have free (or any) automorphisms, but a manifold with concrete moduli can have free group actions up to $\ZZ_5 \times \ZZ_5$. In particular, this is the free automorphism group of the Fermat quintic. No larger free automorphism group can exist because its order must divide the Euler characteristic $\chi=-200=-2^3 5^2$, and free actions cannot have order $2$ because they would descend from the permutation group $S_5$ and will necessarily have fixed points.

Quotienting the manifold by its free group $\Gamma$ creates a new smooth manifold, and in the case of highly symmetric quotients the symmetries are often preserved \cite{mishra_calabi-yau_2017}. For example, the symmetries of quintic quotients are classified, and the largest symmetry group the resulting manifolds can have is $\mathrm{Dic}_5$, a dicyclic group of order $20$ \cite{candelas_highly_2018}. Such analysis could be extended to CICYs. While this has not yet been done, discovering new families of highly symmetric CICY quotients could have interesting implications for phenomenology.

\subsection{Flat Metric Isometries vs Manifold Automorphisms}

The Calabi--Yau manifolds are distinguished by having a unique Ricci-flat metric (in each K\"ahler class) \cite{yau_ricci_1978}. These flat metrics $g_\mathrm{flat}$ have isometry groups $\Isom\parens*{X,g_\mathrm{flat}}$. Let us focus on the holomorphic subgroup $\Isom_\mathrm{hol}\parens*{X, g_\mathrm{flat}}$ for now. Isometry groups are clearly related to the automorphisms of the underlying manifold, and one has
\begin{equation*}
    \Isom_\mathrm{hol}\parens*{X, g_\mathrm{flat}} \le \Aut(X)
\end{equation*}
by definition.

We can recall that quotienting by free groups $\Gamma$ will create a smooth Calabi--Yau, which itself has a unique Ricci-flat metric. Since its pullback must coincide with the unique Ricci-flat metric on $X$, and since the metric is naturally $\Gamma$-invariant, we conclude that elements of $\Gamma$ are isometries of $g_\mathrm{flat}$. Therefore, we arrive at a lower bound
\begin{equation*}
    \Gamma \le \Isom_\mathrm{hol}\parens*{X, g_\mathrm{flat}}.
\end{equation*}
It is not obvious from this approach whether a better bound can exist. Certainly, the current approach cannot produce a stronger lower bound.

Nevertheless, the bound can indeed be significantly sharpened, as we will show. We begin with a quote from Donaldson:
\begin{quote}
    A compact Riemann surface admits a metric of constant curvature,
unique up to the action of the holomorphic automorphisms of the
surface \cite{donaldson_scalar_2001}.
\end{quote}
That is, for \emph{surfaces}, absolute uniqueness of the constant curvature metric would imply that all automorphisms are its isometries. This intuition translates to our case via Yau's theorem, where if an automorphism preserves K\"ahler class it must immediately be an isometry of the flat metric by uniqueness.

Formally, let $\omega$ be the K\"ahler form associated to $g_\mathrm{flat}$, and let $\Aut(X,[\omega]) \le \Aut(X)$ be the subgroup that preserves K\"ahler class $\bracks{\omega}$. Furthermore, let $f\in\Aut\parens*{X,\bracks*{\omega}}$. Since $\Ric$ is a \emph{natural} tensor it commutes with pullbacks \cite{lee_curvature_2018}, and we have 
\begin{equation*}
    \Ric\parens*{f^* g_\mathrm{flat}}=f^* \Ric\parens*{g_\mathrm{flat}}=f^*(0)=0.
\end{equation*}
Since $f$ preserves K\"ahler class we also have $\bracks*{f^* \omega}=\bracks*{\omega}$. Hence, $g_\mathrm{flat}$ and $f^* g_\mathrm{flat}$ are two flat metrics in the same class. Therefore, they must be equal to each other by uniqueness, so by definition $f\in\Isom_\mathrm{hol}\parens*{X, g_\mathrm{flat}}$. In the other direction, every isometry trivially preserves K\"ahler forms. Therefore, we have 
\begin{equation*}
\Isom_\mathrm{hol}\parens*{X, g_\mathrm{flat}} = \Aut\parens*{X, \bracks*{\omega}}.    
\end{equation*}
 Crucially, and this is easily missed, this argument does not require the group action to be free. Furthermore, we treated $\omega$ and $g_\mathrm{flat}$ as global objects, and thus we never had to introduce coordinate patches on $X$ or look at how they transform on their overlaps. Since they are well-defined globally, they must behave as expected in any concrete representation.
 
 To completely describe the isometry group and account for anti-holomorphic actions that might or might not exist, we can write
\begin{equation}\label{eq:isom_ses}
   \Isom\parens*{X, g_\mathrm{flat}}/\Aut\parens*{X, \bracks*{\omega}} \cong G,
\end{equation}
where $G=\ZZ_2$ if $X$ has anti-holomorphic symmetries, and $G=1$ otherwise. That is, the holomorphic automorphisms form a normal subgroup of index at most two inside the full isometry group.

This holds because on a Calabi--Yau with strict holonomy $\SU(n)$, every isometry either commutes with complex structure $J$ (holomorphic), or sends it to $-J$ (anti-holomorphic), and there are no other possibilities. Let $\phi\colon \Isom\parens*{X, g_\mathrm{flat}} \to \ZZ_2$ send holomorphic isometries to $0$ and anti-holomorphic ones to $1$. By the First Isomorphism Theorem we have
\begin{equation*}
    \Isom\parens*{X, g_\mathrm{flat}}/\ker(\phi) \cong \im(\phi)
\end{equation*}
Since we have $\ker(\phi) = \Isom_\mathrm{hol}\parens*{X, g_\mathrm{flat}} =  \Aut\parens*{X, \bracks*{\omega}}$, and $\im(\phi)=G$, we are immediately done. If the holonomy is not exactly $\SU(n)$ but a proper subgroup, $G$ can have a much more interesting structure. For example, on K3 surfaces, there are additional isometries arising from Hyperk\"ahler transformations, and flat tori have an enormous isometry group.

\subsection{Isometries on Fermat Calabi--Yau}

To summarise, we revisit the Fermat quintic, our model example. 

Firstly, we can improve upon equation \eqref{eq:isom_ses} because there is a natural $\ZZ_2$ action given by conjugation, and write $\Isom\parens*{X, g_\mathrm{flat}} = \Aut\parens*{X, \bracks*{\omega}} \rtimes \ZZ_2$. It is not possible to do this always, though.

Secondly, when learning the flat metric we typically parameterise it as $g_\mathrm{flat} \approx \iota^* g_\mathrm{ref} + \partial\overline{\partial} \phi_\mathrm{NN}$, so we are actually interested in the symmetries of $\phi$. To ensure that $\phi$ has a large symmetry group, we should pick a highly symmetric $g_\mathrm{ref}$. Specifically, if $g_\mathrm{ref}$ has more symmetries than $g_\mathrm{flat}$, we can infer the symmetries of $\phi$ from $\Isom\parens*{X, g_\mathrm{flat}}$. For this reason, using the pullbacked Fubini--Study metric as a reference is a good choice.

To see how to obtain the symmetries of $\phi$ in practice, let us assume $\Isom\parens*{X, g_\mathrm{flat}} \le \Isom\parens*{\PP^n, g_\mathrm{ref}}$. Then, we can average over $\Aut$ using the Reynolds operator to obtain
\begin{equation*}
    \mathcal{R}_{\Isom\parens*{X, g_\mathrm{flat}}}\parens*{g_\mathrm{flat}} = \mathcal{R}_{\Isom\parens*{X, g_\mathrm{flat}}}\parens*{g_\mathrm{ref} + \partial\overline{\partial} \phi}.
\end{equation*}
After using the fact that $\mathcal{R}$ commutes with everything here, and that the two metrics are invariant under the group, we obtain
\begin{equation*}
    g_\mathrm{flat} =g_\mathrm{ref} + \partial\overline{\partial}  \mathcal{R}_{\Isom\parens*{X, g_\mathrm{flat}}}\parens*{\phi}.
\end{equation*}
Therefore we can apply symmetries of $g_\mathrm{flat}$ when learning $\phi$. However, this does not mean that one ought to sum over all $\abs{\Isom\parens*{X, g_\mathrm{flat}}}=75{,}000$ transformations in each forward pass, as that would be computationally infeasible. Rather, the lesson is that these symmetries should be built into the neural architecture that learns $\phi$.

\section{Point Sampling}\label{sec:4}
Once we have fixed a Calabi--Yau variety by specifying its defining polynomial(s), we can begin to learn differential forms and tensors such as the Ricci-flat metric. However, doing so requires working with points on the manifold, so we require a way to sample them.

Ideally, we would like to sample points uniformly with respect to the intrinsic volume form $\Omega$ of the Calabi--Yau, but this is not possible. One solution is to use rejection sampling to get points close to the manifold, but this is both inaccurate and expensive. Instead, the commonly adopted solution in practice is to sample points in the ambient space, and then use a theorem of Shiffman and Zelditch \cite{shiffman_distribution_1999} to find intersections with the Calabi--Yau.

This results in a biased sample on the manifold, but with a reference distribution that we can control. Therefore, we are able to compute the integration weights
\begin{equation*}
w = \frac{\mathrm{dVol}_\Omega}{\mathrm{dVol}_\mathrm{ref}} = \frac{\det g_\mathrm{flat}}{\det g_\mathrm{ref}},
\end{equation*}
which are necessary for calculating global objects on the manifold, such as integrals relating to the flat metric. Here, the ratio of volume elements is exactly the ratio between determinants of the Ricci-flat and reference metrics. The latter is computable because we have control over it, but crucially, $\det g_\mathrm{flat}$ is also computable even if $g_\mathrm{flat}$ itself is not. This is because we can obtain it directly from the Monge--Amp\`ere equation $\det g_\mathrm{flat} = \kappa \Omega \wedge \overline\Omega$, where $\Omega$ is again the volume form, and $\kappa$ is a real constant. The right-hand side is computable because $\Omega$ can be derived directly from the defining polynomial using the Poincar\'e residue, and $\kappa$ can be estimated numerically.

\subsection{Integration Weights on Hypersurfaces in \texorpdfstring{$\PP^n$}{Complex Projective Plane}}
To see how this works in practice, let us turn to Calabi--Yau hypersurfaces with inclusion map $\iota\colon X \hookrightarrow \PP^n$. We denote the homogeneous coordinates with $Z_i$, and local coordinates on a patch of $\PP^n$ where $Z_0=1$ with $z_i$. Here, the ambient metric is the Fubini--Study metric
\begin{equation*}
g_{ij} = \delta_{ij}/\parens[\big]{1 + \norm{z}^2} - \overline{z}_iz_j/\parens[\big]{1 + \norm{z}^2}{}^2,    
\end{equation*}
and we can naturally sample according to it by constructing $\PP^n$ as a quotient $S^{2n+1}/S^1$ and pulling back the round metric from the sphere.

As a concrete example, recall the Fermat hypersurface given by
\begin{equation*}
    Q\colon Z_0^{n+1}+\dots+Z_n^{n+1}=0.
\end{equation*}
Fixing a patch, e.g.~$Z_0=1$, we obtain the defining equation
\begin{equation*}
    1+z_1^{n+1}+\dots+z_n^{n+1}=0.
\end{equation*}
To obtain local coordinates on $X$ we can eliminate one coordinate from $\PP^n$, e.g.~$z_n$. Next, we compute the volume form
\begin{equation*}
    \Omega = \frac{1}{\abs{z_n}^{2n}}.
\end{equation*}
The pullbacked Fubini--Study metric $\iota^* g_\mathrm{FS}$ equals $J g_\mathrm{FS} J^\dag$, where $J$ is the Jacobian matrix resulting from elimination of $z_n$. Finally, we have everything needed to compute the integration weights $w = \det g_\mathrm{flat} / \det \iota^* g_\mathrm{FS}$.

While the above derivation is mathematically correct, it is arguably unsatisfactory. Although the weights $w$ are well-defined globally and are coordinate-independent, this is not immediately obvious from their definition, and computing them required choosing a local patch and introducing a basis. It is natural to ask whether that was necessary, or if a more direct computation is available.

This is addressed for hypersurfaces in $\PP^n$ \cite{mirjanic_symbolic_2025}, where one has
\begin{equation*}
    \frac{\det g_\mathrm{flat}}{\det \iota^* g_\mathrm{FS}} = \frac{\norm{Z}^{2n}}{\norm{\nabla Q}^2}
\end{equation*}
for any defining polynomial $Q$ in $\PP^n$, where $\nabla Q$ is the Euclidean gradient with respect to $Z$, and $\norm{-}$ is the usual vector norm. This holds up to a constant scaling factor that is omitted here due to different normalising conventions, and because weights will be normalised to sum to one anyway.

The advantage of having this closed formula for $w$ is that it is now obvious that the weights are well-defined scalars that are invariant under homogeneous transformations $Z\mapsto \lambda Z$. Furthermore, the new formula is \emph{manifestly global}, unlike the original definition. This is useful for pedagogical purposes, and simplifies reasoning about the properties of the weights.

Consider the aforementioned Fermat hypersurface. There, the weights can be directly computed to be
\begin{equation*}
    w = \frac{\parens*{\abs{Z_0}^2 + \dots + \abs{Z_n}^2}^n}{\abs{Z_0}^{2n}+\dots+\abs{Z_n}^{2n}}.
\end{equation*}
Comparing invariants of $w$ to isometries of the Fermat Calabi--Yau, we easily observe that $w$ is invariant to both permutations and toric actions. In fact, $w$ has a larger symmetry group since it is invariant under arbitrary phase shifts. Meanwhile, this $U(1)^{n+1}$ symmetry is not a symmetry of the underlying manifold.

Is there a deeper meaning in $w$ having more symmetries than the manifold itself? One proposition \cite{mirjanic_symbolic_2025} is that this enlarged symmetry group could have opportunities to manifest in other differential forms on the manifold. This could also be related to mirror symmetry \cite{strominger_mirror_1996}, since resolving conifold singularities on the quotient $X/\ZZ_{n+1}^{n+1}$ will result in a Calabi--Yau that is the mirror of the one defined by $\prod Z_i = 0$, and the latter does have the desired $U(1)^{n+1}$ symmetry group. However, this is all speculative, and the formula remains a curiosity for now.

\subsection{Integration Weights on Complete Intersections}
A natural follow-up question is whether a similar formula exists for general CICYs. Here we present a novel, manifestly global formula for this volume ratio.
\begin{theorem}
\label{thm:closed_form_w_cicy}
    Consider a Complete Intersection Calabi--Yau (CICY) manifold defined by $m$ multihomogeneous polynomials $Q_1, \dots, Q_m$ embedded into ambient space $\mathcal{A} = \PP^{n_1} \times \dots \times \PP^{n_k}$ with K\"ahler moduli $t_1,\dots, t_k$. Then,
    \begin{equation*}
        \frac{\mathrm{dVol_{flat}}}{\mathrm{dVol_{FS}}} = \frac{\det g_\mathrm{flat}}{\det \iota^* g_\mathrm{FS}} = \frac{\prod_{r=1}^k t_r^{-n_r} \norm*{Z^{(r)}}^{2(n_r+1)}}{\det_{ij}\parens*{ \sum_{r=1}^k t_r^{-1} \norm{Z^{(r)}}^2 \innerprod*{\nabla^{(r)} Q_i, \nabla^{(r)} Q_j}}}
    \end{equation*}
    where $Z^{(r)}$ represents the vector of homogeneous coordinates for the $r$-th projective space $\PP^{n_r}$, and $\nabla^{(r)} Q_i$ are standard Euclidean gradients with respect to the coordinates $Z^{(r)}$.
\end{theorem}
There are multiple pedagogical benefits to having an explicit formula for this ratio. Firstly, just as with the hypersurface variant, it is easy to observe that this is a globally well-defined function on the manifold. Furthermore, we can observe that this is a proper generalisation of the original formula, which reduces to it when there is only one polynomial and projective space. We also observe that while the formula also involves taking the determinant, there are fewer polynomials $Q$ than total homogeneous variables ($m < \sum (n_i + 1)$), which makes our formula more efficient. 

While the formula is degree-$0$ homogeneous in the coordinates $Z$, it is also homogeneous in the K\"ahler moduli $t$. Specifically, if $t\mapsto \lambda t$, then the numerator scales like $\lambda^{\sum n_i}$. On the other hand, the denominator computes the determinant of an $m\times m$ matrix, so it scales like $\lambda^m$. Since $\sum n_i-m=\dim X$, we conclude that the volume ratio is homogeneous of degree $\dim X$ in the moduli $t$. This perfectly matches $\det \iota^* g_\mathrm{FS}$ being degree $\dim X$ homogeneous in the moduli, while $\det g_\mathrm{flat}$ is invariant.

CICYs are typically specified by a configuration matrix that keeps track of degrees of polynomials $Q_i$ with respect to each projective space. For example, a general Tian--Yau manifold has the configuration 
\begin{equation*}
\bracks*{
\begin{array}{c|ccc}
\PP^3 & 3 & 0 & 1 \\
\PP^3 & 0 & 3 & 1
\end{array}
},
\end{equation*}
and a concrete manifold can be defined with polynomials
\begin{equation}\label{eq:tian-yau}
  \begin{aligned}
    Q_1&\colon Z_0^3 + Z_1^3 + Z_2^3 + Z_3^3 = 0, \\
    Q_2&\colon W_0^3 + W_1^3 + W_2^3 + W_3^3 = 0, \\
    Q_3&\colon Z_0 W_0 + Z_1 W_1 + Z_2 W_2 + Z_3 W_3 = 0,
\end{aligned}  
\end{equation}

where $Z$ and $W$ are homogeneous coordinates from different $\PP^3$s. We observe that the configuration matrix is not directly encoded into the formula for the volume ratio. Rather, it is present via $\nabla^{(r)} Q_i$, which already includes all information about relationship between polynomials and projective spaces.

This formula may at first seem like an excellent way to simplify computation of point weights. However, the problem is that this ratio no longer equals the integration weights, because the reference measure is no longer the ambient Fubini--Study measure. This happens due to the way the point sampling is implemented using the Shiffman and Zelditch's theorem \cite{shiffman_distribution_1999}, and results in subtle and interesting consequences.

Namely, for each polynomial, a concrete coordinate that will be eliminated has to be chosen. Since there are three polynomials $Q$ and only two projective spaces, one must necessarily break symmetry. Without loss of generality we will eliminate two coordinates from the second projective space and only one from the first. Thus, if $J_1$ and $J_2$ are the K\"ahler forms, the Fubini--Study measure corresponds to the top form $(J_1+J_2)^3$ and is symmetric, while the reference measure chosen by the sampling algorithm would be $J_1\wedge J_1 \wedge J_2$, reflecting the broken symmetry of the ambient space.

To see how this symmetry breaking appears in practice, we look at qualitative properties of sampled points like the histogram of norms of coordinate vectors. In Figure \ref{fig:tian_yau_norms}, we compare the norms of coordinate vectors for first and second $\PP^3$ for two different measures, sampled $d\mathrm{Vol}_{\mathrm{ref}}$, and Ricci-flat $d\mathrm{Vol}_\Omega$. Note that the points are always normalised in such a way that $\max Z_i=1$, so $2\le\norm{Z}^2\le 4$, and similarly for $W$. Furthermore, the plots are normalised so that the total volume equals $1$ in both cases.

\begin{figure}[ht]
    \centering
    \includegraphics{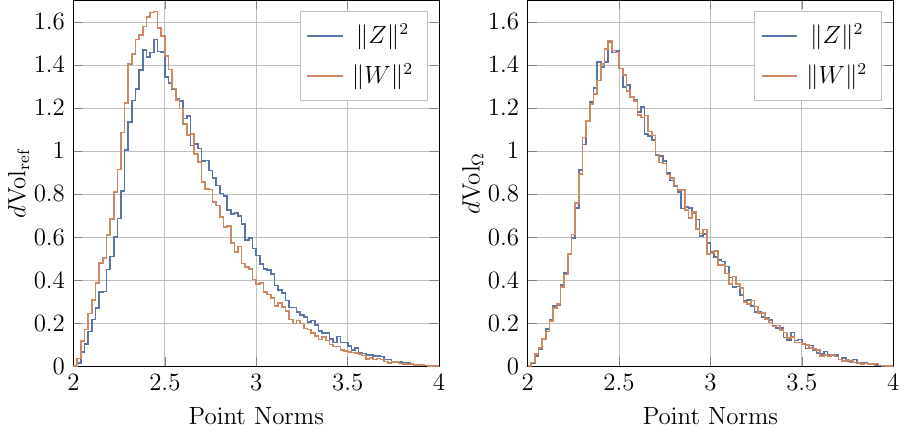}
    \caption{Numerical histogram of point norms on the Tian--Yau manifold (\ref{eq:tian-yau}), with respect to the reference measure (left), and the Ricci-flat measure (right). The latter is symmetric with respect to the two projective spaces, while the former is not.}
    \label{fig:tian_yau_norms}
\end{figure}

We observe that the distributions are different under the sampling measure, which would not be the case if it were symmetric like Fubini--Study. Meanwhile, the two histograms become identical when adjusted for the flat metric, showing that the integration weights properly cancelled out the sampling asymmetry.

This behaviour is somewhat inelegant, and one could ask if there is an inexpensive way to avoid this. However, this is more easily said than done. Since the sampling strategy requires using a concrete top form like $J_1\wedge J_1 \wedge J_2$, in order to make it symmetric we would need to randomly switch between all terms in the expansion of $(J_1 + J_2)^3$, properly accounting for coefficients and K\"ahler moduli. This quickly becomes unwieldy as we move to CICYs made out of more than two projective spaces.

On the other hand, perhaps we could change the sampling strategy itself to make it inherently symmetric. We can recall the Segre embedding \cite{hartshorne_algebraic_1977}
\begin{equation*}
    \PP^{n_1}\times\PP^{n_2}\times\dots\times\PP^{n_k} \hookrightarrow \PP^{(n_1+1)(n_2+1)\dots(n_k+1)-1}
\end{equation*}
and the fact that pullbacks of the Fubini--Study metric under it are itself Fubini--Study metrics. Thus, one could try to work directly in this single projective space instead of the original product. However, this is unfeasible in practice since the number of homogeneous coordinates explodes: $\prod (n_i+1) - 1 \gg \sum n_i$.

Finally, one can ask if it is possible to change the formula to match the reference measure. However, this is not feasible either, and would involve replacing the determinant in the denominator with a complex alternating sum that depends on the exact top form used by the sampler.

In the end, the best thing we can do is to proceed to the proof of Theorem \ref{thm:closed_form_w_cicy}.

\subsection{Proof of Theorem \ref{thm:closed_form_w_cicy} (on CICY Volume Ratios)}
Our strategy will be to avoid computing determinants by hand and instead manipulate them with linear algebra lemmas such as Schur complement. Since we are working in local coordinates, we will use Euler's identity to reintroduce homogeneous coordinates.
\begin{lemma}[Schur Complement]
    Let $M = \begin{pmatrix} A & B \\ C & D \end{pmatrix}$ be a block matrix with $A$ and $D$ square. 
    If $A$ is invertible, then $\det(M) = \det(A) \det(D - C A^{-1} B)$.
    Similarly, if $D$ is invertible, then $\det(M) = \det(D) \det(A - B D^{-1} C)$. Furthermore, there are known formulas for $M^{-1}$ in these cases.
\end{lemma}
\begin{lemma}[Euler's Homogeneous Function Theorem]
    Every degree $n$ homogeneous function $Q$ in variables $Z_i$ satisfies
    \begin{equation*}
        \sum_i Z_i \frac{\partial Q}{\partial Z_i} = n Q.
    \end{equation*}
\end{lemma}

Let $N = \sum n_i=\dim\mathcal{A}$. The metric on $\mathcal{A}$ is block-diagonal and composed of Fubini--Study metrics scaled by the moduli. Thus, $g_\mathrm{FS} = \bigoplus_{r=1}^k t_r g_\mathrm{FS}^{(r)}$.

Without loss of generality, we pick a patch where $Z_0^{(r)}=1$ for all $r$, and denote local coordinates with $z^{(r)}$. In total, there are $N$ coordinates in $z$.

Next, let $D_{a, I} = \partial Q_a / \partial z_I$ be the $m \times N$ Jacobian matrix of the defining polynomials. We partition our coordinates $z$ into $N-m$ independent coordinates $x$ and $m$ dependent coordinates $y$. Similarly, we partition the Jacobian into $D = \begin{pmatrix} D_x & D_y \end{pmatrix}$, where $D_y$ is an invertible $m \times m$ matrix.
    
The inclusion map $\iota\colon X \hookrightarrow \mathcal{A}$ locally takes the form $x \mapsto (x, y(x))$. Taking the total derivative of the constraints $Q_i(x, y(x)) = 0$ yields $D_x + D_y \frac{\partial y}{\partial x} = 0$, implying $\frac{\partial y}{\partial x} = -D_y^{-1} D_x$.

Finally, we can now give the classical definition of pullback Jacobian $J$, an $(N-m) \times N$ matrix such that $\iota^* g_\mathrm{FS} = J g_\mathrm{FS} J^\dag$, with
\begin{equation*}
    J = \begin{pmatrix} I_{N-m} & -D_x^\dag (D_y^\dag)^{-1} \end{pmatrix}.
\end{equation*}
Note that by construction we have $J D^\dag = D_x^\dag - D_x^\dag (D_y^\dag)^{-1} D_y^\dag = 0$.
\begin{lemma}[\textcite{berglund_residue_1995}]
    The top form $\Omega$, as obtained via the Poincar\'e residue formula, takes the form
    \begin{equation*}
        \Omega = \frac{dx_1 \wedge \dots \wedge dx_n}{\det D_y}
    \end{equation*}
    in our chosen coordinate patch. The corresponding flat volume form is proportional to $\Omega \wedge \bar{\Omega}$, which yields:
    \begin{equation*}
        \det g_\mathrm{flat} = \frac{1}{\abs{\det D_y}^2}
    \end{equation*}
\end{lemma}
\begin{lemma}
    The inverse of the scaled Fubini-Study metric on a single $\PP^{n_r}$ is given by
    \begin{equation*}
        (t_r g_\mathrm{FS}^{(r)})^{i\bar{j}} = \frac{1}{t_r} \parens*{1+\norm{z^{(r)}}^2}\parens*{\delta^{ij} + z_i^{(r)} \bar{z}_j^{(r)}}
    \end{equation*}
    The determinant of the ambient metric $g_\mathrm{FS}$ is the product of the determinants of the individual blocks, so
    \begin{equation*}
        \det g_\mathrm{FS} = \prod_{r=1}^k \frac{t_r^{n_r}}{\parens*{1+\norm{z^{(r)}}^2}^{n_r+1}} = \prod_{r=1}^k t_r^{n_r} \norm{Z^{(r)}}^{-2(n_r+1)}
    \end{equation*}
\end{lemma}
Now, while computing $\det g_\mathrm{FS}$ was easy, the main difficulty lies in computing $\det \iota^* g_\mathrm{FS}$.
\begin{claim}\label{claim:1}
    $\det \iota^* g_\mathrm{FS} = \det g_\mathrm{FS} \det\parens*{D g_\mathrm{FS}^{-1} D^\dag} / \abs{\det D_y}^2$
\end{claim}
\begin{proof}
    Let
    \begin{equation*}
        M = \begin{pmatrix} g_\mathrm{FS} & D^\dag \\ D & 0 \end{pmatrix},
    \end{equation*}
    be an $(N+m)\times(N+m)$ matrix, for reasons that will soon be apparent. Using Schur complement we compute
    \begin{equation}
    \label{eq:schur1_kahler}
        \det M = \det g_\mathrm{FS} \det\parens*{-D g_\mathrm{FS}^{-1} D^\dag} = (-1)^m \det g_\mathrm{FS} \det\parens*{D g_\mathrm{FS}^{-1} D^\dag}.
    \end{equation}
    Next, we apply a basis transformation to $M$ using 
    \begin{equation*}
        S = \begin{pmatrix} I_{N-m} & -D_x^\dag (D_y^\dag)^{-1} & 0 \\ 0 & I_m & 0 \\ 0 & 0 & I_m \end{pmatrix},
    \end{equation*}
    to obtain $\tilde{M} = S M S^\dag$. Note that $S$ is upper-block triangular, so we immediately have $\det S = \det I_{N-m} (\det I_m)^2 = 1$, and therefore $\det M = \det \tilde{M}$. Furthermore, note that the top row of $S$ is identical to $J$. By introducing $E = \begin{pmatrix} 0 & I_m \end{pmatrix}$ from the second row of $S$, and explicitly computing $\tilde{M}$, we find
    \begin{equation*}
        \tilde{M} = \begin{pmatrix}
            J g_\mathrm{FS} J^\dag & J g_\mathrm{FS} E^\dag & J D^\dag \\
            E g_\mathrm{FS} J^\dag & E g_\mathrm{FS} E^\dag & E D^\dag \\
            D J^\dag & D E^\dag & 0
        \end{pmatrix} = \begin{pmatrix} 
            J g_\mathrm{FS} J^\dag & J g_\mathrm{FS} E^\dag & 0 \\
            E g_\mathrm{FS} J^\dag & E g_\mathrm{FS} E^\dag & D_y^\dag \\
            0 & D_y & 0
        \end{pmatrix},
    \end{equation*}
    where we used $JD^\dag=0$, $DE^\dag=D_y$ to simplify the expression. Letting
    \begin{align*}
    K &= \begin{pmatrix} 
        E g_\mathrm{FS} E^\dag & D_y^\dag \\
        D_y & 0
    \end{pmatrix} \\
    L &= \begin{pmatrix} 
        J g_\mathrm{FS} E^\dag & 0
    \end{pmatrix},
    \end{align*}
    be the components of $\tilde{M}$, we find 
    \begin{equation*}
        K^{-1} = \begin{pmatrix}
            0 & D_y^{-1} \\
            \parens*{D_y^\dag}^{-1} & -\parens*{D_y^\dag}^{-1}E g_\mathrm{FS}E^\dag D_y^{-1}
        \end{pmatrix}
    \end{equation*}
    because all the blocks are invertible. Since the top left block of $K$ is zero, we immediately have $LKL^\dag=0$. Thus, we are able to compute $\det \tilde{M}$ by taking the Schur complement of $K$ as
    \begin{equation}
    \begin{aligned}
    \label{eq:schur2_kahler}
        \det \tilde{M} &= \det(K) \det\parens*{J g_\mathrm{FS} J^\dag - L^\dag K^{-1}L}\\ 
        &= \det(K) \det\parens*{J g_\mathrm{FS} J^\dag} = \parens*{(-1)^m \abs{\det D_y}^2} \det \iota^* g_\mathrm{FS}.
    \end{aligned}
    \end{equation}
    Equating \eqref{eq:schur1_kahler} and \eqref{eq:schur2_kahler} completes the proof.
\end{proof}

\begin{claim}\label{claim:2}
    $\parens*{D g_\mathrm{FS}^{-1} D^\dag}_{ab} = \sum_{r=1}^k 1/t_r \norm{Z^{(r)}}^2 \innerprod*{ \nabla^{(r)} Q_a, \nabla^{(r)} Q_b}$
\end{claim}
\begin{proof}
    Directly expanding the $(a,b)$-th entry of $D g_\mathrm{FS}^{-1} D^\dag$ we obtain
    \begin{align*}
        \parens*{D g_\mathrm{FS}^{-1} D^\dag}_{ab} &= \sum_{r=1}^k \sum_{i,j=1}^{n_r} \frac{\partial Q_a}{\partial z_i^{(r)}} (t_r g_\mathrm{FS}^{(r)})^{i\bar{j}} \overline{ \frac{\partial Q_b}{\partial z_j^{(r)}} } \\
        &= \sum_{r=1}^k \frac{1}{t_r}\parens*{1+\norm*{z^{(r)}}^2} \parens*{ \sum_{i=1}^{n_r} \frac{\partial Q_a}{\partial z_i^{(r)}} \overline{\frac{\partial Q_b}{\partial z_i^{(r)}}} + X},
    \end{align*}
    where
    \begin{equation*}
        X = \parens*{ \sum_{i=1}^{n_r} z_i^{(r)} \frac{\partial Q_a}{\partial z_i^{(r)}} } \overline{\parens*{ \sum_{j=1}^{n_r} z_j^{(r)} \frac{\partial Q_b}{\partial z_j^{(r)}} }}.
    \end{equation*}
    By Euler, since $Q_a$ is both homogeneous and vanishing, we have $\sum_{I=0}^{n_r} Z_I^{(r)} \frac{\partial Q_a}{\partial Z_I^{(r)}} = 0$. Since we are working in the patch $Z_0^{(r)} = 1$, we divide by $Z_0^{(r)}$ to obtain
    \begin{equation*}
        \sum_{i=1}^{n_r} z_i^{(r)} \frac{\partial Q_a}{\partial z_i^{(r)}} = - \frac{\partial Q_a}{\partial Z_0^{(r)}}.
    \end{equation*}
    Thus, we find that $X = \partial Q_a / \partial Z_0^{(r)} \overline{\partial Q_b / \partial Z_0^{(r)}}$. Hence, we naturally restored the standard Euclidean dot product over all $n_r+1$ homogeneous coordinates, as desired.
\end{proof}

Finally, we can put everything together. Taking the ratio of the volume forms and utilizing claim \ref{claim:1} we obtain
    \begin{equation*}
        \frac{\det g_{\text{flat}}}{\det \iota^* g_\mathrm{FS}} = \frac{\frac{1}{\abs{\det D_y}^2}}{\frac{\det g_\mathrm{FS} \det\parens*{D g_\mathrm{FS}^{-1} D^\dag}}{\abs{\det D_y}^2}} = \frac{1}{\det g_\mathrm{FS} \det\parens*{D g_\mathrm{FS}^{-1} D^\dag}}
    \end{equation*}
    Expanding the determinants with claim \ref{claim:2} completes the proof of the theorem.

\section{Learning Ricci-flat Metrics with GNNs}\label{sec:5}

Having fixed the Calabi--Yau manifold with a specific symmetry group and sampled the points, we proceed to learn the Ricci-flat metric in \emph{cymyc} \cite{berglund_cymyc_2024}.

There are different architectures to choose from, and each encodes different geometric properties and constraints of the flat metric into the neural network. We have already seen the PhiModel \cite{larfors_learning_2021} for example, which parameterises the flat metric using a single scalar function $\phi$ instead of learning metric components individually, with $g_\mathrm{flat} \approx \iota^* g_\mathrm{ref} + \partial\overline{\partial}\phi_\mathrm{NN}$. This ensures that the learnt metric is K\"ahler by construction. A variation involves adding the spectral layer \cite{berglund_machine_2023}, where the \emph{spectral features} $Z_i \overline{Z_j} / \norm{Z}^2$ are used as neural network inputs instead of homogeneous coordinates $Z_i$ themselves. The purpose of these features is to make the network invariant to the $\CC^*$ projective symmetry, and thus make $\phi$ well defined on the projective ambient space. In fact, it is known that these features are expressive enough for the network to learn any $\CC^*$-symmetric function \cite{villar_scalars_2021,weyl_classical_1946}.

A machine learning discipline that studies symmetry invariant or equivariant neural networks is Geometric Deep Learning \cite{bronstein_geometric_2021}. Within this framework, Graph Neural Networks (GNNs) are used to introduce permutation symmetry into the architecture, making the network invariant to different ordering of nodes and edges in the input graph. Various extensions to different symmetry groups and manifolds also exist. Despite the shared focus on symmetry and representation, applications of GNNs to metric learning are surprisingly under-explored. 

One way to incorporate GNNs could be to discretise the manifold into a graph and treat points as nodes. Then, sheaf networks could be used to represent different differential forms on the nodes, one of them being the learnt metric. However, it is unclear if this approach can properly capture properties like holonomy of the manifold, so it might not be sufficiently expressive for the task.

The approach that we adopt is to apply the GNN pointwise on a synthetic graph in which nodes represent ambient coordinates. Then the graph output is aggregated to produce a scalar value that represents $\phi$ in the PhiModel. The purpose of the GNN here is to capture the permutation symmetry that is common on Calabi--Yau manifolds. Furthermore, this approach is compatible with spectral networks, since we can treat $Z_i \overline{Z_j} / \norm{Z}^2$ as input features on graph edges.

We test this architecture on the Fermat quintic, using a dataset of $10^5$ points. To make the comparison fair, we use spectral networks and GNNs with roughly the same number of parameters (around $8$ thousand), and we use the same hyperparameters in both cases. 

\begin{figure}[ht]
    \centering
    \includegraphics{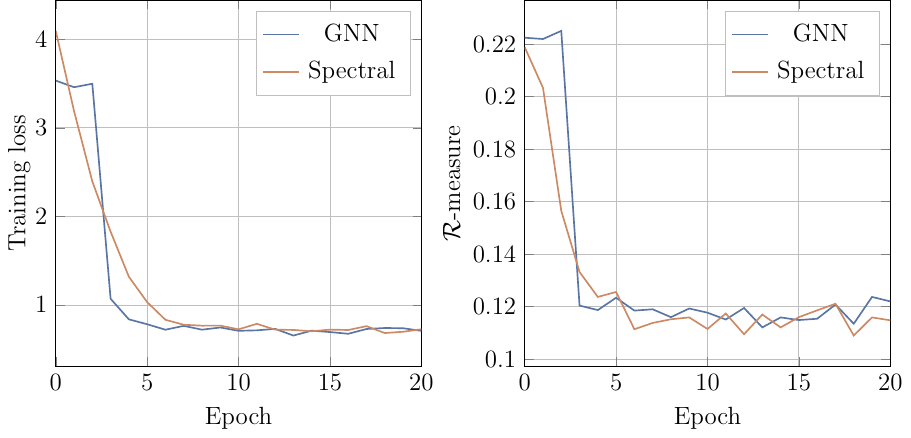}
    \caption{Evolution of losses when learning the Ricci-flat metric on the Fermat quintic.}
    \label{fig:losses}
\end{figure}

Figure \ref{fig:losses} shows that training loss exhibits a much sharper transition from high loss to low loss regime when GNNs are used. This is likely a consequence of the model being intrinsically symmetric and thus able to learn the metric directly, without needing to also learn the symmetries.

Another approach to introducing symmetries into the metric by construction is by projecting points to fundamental domains of the symmetry group \cite{hendi_learning_2024}. However, this approach exhibited an unexpected behaviour where expanding the symmetry group \emph{increased} some loss measures like $\mathcal{R}$ measure, even though the training loss kept improving. This unexpected behaviour implies that including more symmetries makes the loss landscape more complicated. Practically, this means that there exist \emph{pathological} metrics that appear near-flat according to some measures, while remaining highly curved according to others.

Since the ground truth is unavailable, one must exercise care when designing the architecture and the training pipeline, and flatness should be evaluated with multiple different measures. 

On the other hand, some loss measures can be not just uninformative, but even counter-productive. For example, the Euler characteristic $\chi$ can be both computed topologically, and estimated by integrating the metric. However, any metric should yield the correct answer, so deviations are not a measure of curvature but are a Monte-Carlo error. Similarly, any deviations in volume with respect to the flat metric are also just numerical noise.

In practice, training using the $\sigma$-measure \cite{larfors_learning_2021} seems to work well enough, but one still needs to evaluate some curvature-based metrics like $\norm{\Ric g}$ or the $\mathcal{R}$-measure.

In summary, successful numerical estimation of Ricci-flat metrics greatly benefits from systematically incorporating symmetries into the machine learning pipeline. By deriving the isometry group from first principles, we identified the exact symmetries that can be safely used during training. We also introduced a novel formula for volume ratios on complete intersection Calabi--Yau manifolds, and demonstrated how point sampling can subtly break the symmetries in the dataset. By introducing Graph Neural Networks (GNNs), we captured the symmetries to achieve a sharp phase transition to a low-loss regime while avoiding pathological behaviours.

This has impact on downstream applications of Ricci-flat metrics. For example, symbolic approximations to the Ricci-flat potentials \cite{mirjanic_symbolic_2025,constantin_calabi-yau_2026,lee_approximate_2025} can benefit from the symmetry constraints by enforcing them into their formulae.

Furthermore, having more informed approximations can help us understand how the Ricci-flat metrics transform under manifold deformations. This can lead to better understanding of the string theory landscape. For instance, Reid's fantasy \cite{reid_moduli_1987} famously conjectures that all Calabi--Yau threefolds are connected by extremal transitions, and recent progress heavily supports this \cite{anderson_elliptic_2026}. Thus, AI for metric learning is not merely a tool for localized physical predictions, but a vital instrument for testing global landscape conjectures.

\printbibliography

\end{document}